\documentclass[lettersize,journal]{IEEEtran}
\IEEEoverridecommandlockouts
\usepackage{amsmath,amsfonts}
\usepackage{mathtools}
\usepackage{amsthm}
\usepackage{array}
\usepackage{textcomp}
\usepackage{stfloats}
\usepackage{url}
\usepackage{microtype}
\usepackage{verbatim}
\usepackage{graphicx}
\usepackage{cite}
\usepackage[colorlinks=true, linkcolor=blue, citecolor=blue, urlcolor=blue]{hyperref}
\usepackage{enumitem}

\usepackage[linesnumbered,lined,algoruled,commentsnumbered]{algorithm2e}

\usepackage{caption}
\usepackage{subcaption}
\usepackage{stfloats}
\usepackage{lipsum} 
\usepackage[font=small,skip=0pt]{caption}
\usepackage[export]{adjustbox}
\usepackage[inkscapelatex=false]{svg}

\SetCommentSty{mycommfont}
\theoremstyle{definition}
\newtheorem{theorem}{Theorem}

\usepackage{units}
\newcommand{\nth}[1]{{#1}^{\text{th}}}

\newcommand{\abs}[1]{\left|{#1}\right|}

\newcommand{\red}[1]{{\color{red}{#1}}} 
\newcommand{\NBS}[0]{N_{\mathrm{\scalebox{0.5} {BS} }}}

\newcommand{\RD}[0]{r_{\mathrm{\scalebox{0.5} {RD} }}}

\newcommand{\rf}[0]{r_{\mathrm{\scalebox{0.5} {F}}}}
\newcommand{\BD}[0]{r_{\mathrm{\scalebox{0.5} {BD} }}}
\newcommand{\EBRD}[0]{r_{\mathrm{\scalebox{0.5} {EBRD} }}}

\newcommand{\GR}[0]{\mathcal{G}_\mathrm{\scalebox{0.5} {uca}}}
\newcommand{\RDUCA}[0]{{r}_\mathrm{ \scalebox{0.5} {RD}}^{\scalebox{0.5} {uca}}}
\newcommand{\DUCA}[0]{D_{\scalebox{0.5} {uca}}}
\newcommand{\BDUCA}[0]{{r}_\mathrm{\scalebox{0.5} {BD}}^{\scalebox{0.5} {uca}}}
\newcommand{\RDULA}[0]{{r}_\mathrm{\scalebox{0.5} {RD}}^{\scalebox{0.5} {ula}}}
\newcommand{\DULA}[0]{D_{\scalebox{0.5} {ula}}}
\newcommand{\BDULA}[0]{{r}_\mathrm{ \scalebox{0.5} {BD}}^{\scalebox{0.5} {ula}}}
\newcommand{\GU}[0]{\mathcal{G}_\mathrm{\scalebox{0.5} {ula}}}

\usepackage{acronym}
\input{acronym}

\usepackage{titlesec}
\titlespacing*{\section}{0pt}{.75pt}{.75pt}  
\titlespacing*{\subsection}{0pt}{.5pt}{.5pt}
\titlespacing*{\subsubsection}{0pt}{.15pt}{.15pt}

\begin{document}

\title{Uniform Circular Arrays in Near-Field: Omnidirectional Coverage with Limited Capacity \\

\author{Ahmed Hussain,~\IEEEmembership{Graduate Student Member, IEEE}, Asmaa Abdallah~\IEEEmembership{Senior Member, IEEE},\\ Abdulkadir Celik,~\IEEEmembership{Senior Member, IEEE}, and Ahmed M. Eltawil,~\IEEEmembership{Senior Member, IEEE}

\thanks{Ahmed Hussain, Asmaa Abdallah, and Ahmed M. Eltawil are with KAUST, KSA. Abdulkadir Celik is with University of Southampton, SO17 1BJ UK. (e-mail: ahmed.hussain.2@kaust.edu.sa; 
asmaa.abdallah@kaust.edu.sa; 
a.celik@soton.ac.uk; ahmed.eltawil@kaust.edu.sa).}

}
}

\maketitle

\begin{abstract}
Recent studies suggest that uniform circular arrays (UCAs) can extend the angular coverage of the radiative near-field region. This work investigates whether such enhanced angular coverage translates into improved spatial multiplexing performance when compared to uniform linear arrays (ULAs). To more accurately delineate the effective near-field region, we introduce the effective beamfocusing Rayleigh distance (EBRD)—an angle-dependent metric that bounds the spatial region where beamfocusing remains effective. Closed-form expressions for both beamdepth and EBRD are derived for UCAs. Our analysis shows that, under a fixed antenna element count, ULAs achieve narrower beamdepth and a longer EBRD than UCAs. Conversely, under a fixed aperture length, UCAs provide slightly narrower beamdepth and a marginally longer EBRD. Simulation results further confirm that ULAs achieve a higher sum-rate under the fixed element constraint, while UCAs offer marginal performance gain under the fixed aperture constraint.
\end{abstract}

\section{Introduction}
\IEEEPARstart{F}{uture} wireless networks may operate in the radiative near-field due to the deployment of \ac{UM}-\ac{MIMO} antenna arrays. Unlike the far-field, this region is characterized by spherical wavefronts rather than planar wavefronts. While far-field \acp{UE} can be served based solely on angular separation, the spherical wave propagation in the near-field enables the formation of finite-depth beams in the distance domain. Although the angular beamwidth in radians remains the same in both the near- and far-field regions, the physical beamwidth, which depends on the distance from the antenna array, becomes narrower in the near-field compared to the far-field. These properties enable finer \ac{UE} separation in both the angle and distance dimensions, thereby enhancing spatial multiplexing gains. Recent studies show that both the beamdepth and the limits of beamfocusing depend on the direction in which the beam is focused and the geometry of the antenna array \cite{10934779}. For a \ac{ULA}, the beamdepth is smallest at boresight and increases toward endfire directions \cite{10988573}, suggesting that finite-depth beams—and hence near-field benefits—are confined primarily to \acp{UE} located near boresight. To overcome this restricted angular coverage, \acp{UCA} owing to their rotational symmetry, have been proposed to extend the near-field region \cite{10243590}. This motivates a key question: \textit {Does the omnidirectional coverage of the \ac{UCA} translate into improved spatial multiplexing performance in the near-field?}


In classical electromagnetic theory, the boundary between the near-field and far-field regions is defined by the Rayleigh distance, which is derived based on the phase error between planar and spherical wavefronts. However, the Rayleigh distance overestimates the effective near-field region. To address this limitation, another boundary, referred to as the \ac{ERD}, was introduced for \acp{ULA} in~\cite{10541333} and later extended to \acp{UCA} in~\cite{10243590}. The \ac{ERD} defines the near-field boundary as the distance at which the beamforming loss incurred by adopting the far-field channel model exceeds a predefined threshold. However, the \ac{ERD} also overestimates the effective near-field region in terms of beamfocusing \cite{11428208}. To address this issue, another near-field boundary, termed the \ac{EBRD}, was introduced for \ac{ULA} in~\cite{10988573}, which defines the extent of the near-field region where beamfocusing remains effective. It is worth noting that the effective beamfocused Fraunhofer distance proposed in~\cite{10934779} is equivalent to the \ac{EBRD}.

The reported enhancement in angular coverage for \acp{UCA} compared to \acp{ULA}, based on the \ac{ERD} in~\cite{10243590}, assumes a fixed aperture length and does not account for the elevation coverage inherently provided by \acp{UCA}. Moreover, it is essential to evaluate the achievable communication capacity under practical system constraints, such as a fixed number of antenna elements or a constrained aperture length. To the best of the authors’ knowledge, a rigorous comparison of the near-field multi-user capacity between \ac{UCA} and \ac{ULA}, remains unexplored in the literature. To address these gaps, this work derives closed-form expressions for beamdepth and introduces the \ac{EBRD} for the \ac{UCA}.  A comparative analysis of beamdepth and \ac{EBRD} is carried out for both \ac{ULA} and \ac{UCA} under two practical design constraints: fixed antenna element count and fixed aperture length. Furthermore, simulation results are presented to compare the multi-user sum-rate of \ac{UCA} with \ac{ULA} and \ac{URA} under the same constraints. 
\section{System Model} \label{Sec_II}
We consider the downlink of a single-cell multi-user \ac{UM}-\ac{MIMO} system, where a \ac{BS} with $\NBS$ antennas simultaneously communicates with $K$ single-antenna \acp{UE} within its coverage area. The received signal at the $\nth{k}$ \ac{UE} is
\begin{equation} 
y_{k}=\sqrt{\gamma} \mathbf{w}_{k} \mathbf{h}_{k} s_{k}+\sqrt{\gamma} \sum_{j=1, j \neq k}^{K} \mathbf{w}_{j} \mathbf{h}_{k} s_{j} + {z}_{k}, 
\label{eqn_IIA_1}
\end{equation}
where $\gamma$ is the average \ac{SNR}, $s_{k}$ and $s_{j}$ are the transmit symbols with unit-norm, $\mathbf{w}_{k} \in \mathbb{C}^{1 \times \NBS} $ and $\mathbf{w}_{j} \in \mathbb{C}^{1 \times \NBS}$ are unit norm precoding vectors for \ac{UE} $k$ and $j$, respectively, and $z_{k} \sim \mathcal{CN}(0,1)
$ denotes the zero-mean complex Gaussian additive noise. $\mathbf{h}_{k} \in \mathbb{C}^{\NBS \times 1}$ is the near-field channel vector between the \ac{BS} and the $\nth{k}$ \ac{UE} which is given by
\begin{equation} 
\mathbf{h}_{k}=\sqrt{\NBS}\beta_0\mathbf{b}\left(r_{0},\theta_{0}, \varphi_{0}\right) + \sum_{l=1}^{L}\sqrt{\frac{\NBS}{L}}\beta_l \mathbf{b}\left(r_{l},\theta_{l}, \varphi_{l}\right), 
\label{eqn_IIA_2}
\end{equation}
which contains one \ac{LoS} path ($l=0$) and $L$ N\ac{LoS} paths.  $\mathbf{b} \in \mathbb{C}^{\NBS \times 1}$ denotes the near-field array response vector that is focused at a specific distance $r_{l}$, elevation angle $\theta_{l}$, and azimuth angle $\varphi_{l}$. The channel gain for the LoS path is $\beta_0 = \sqrt{\frac{\kappa}{\kappa+1}}$ while the channel gains for N\ac{LoS} paths follow $\beta_l \sim \mathcal{CN}(0, \sigma_{\beta,l}^2)$, where $\sigma_{\beta,l}^2 = \frac{1}{\kappa+1}$. The Rician factor $\kappa$ represents the power ratio between the \ac{LoS} component and the  N\ac{LoS} components. 

The achievable rate for the $\nth{k}$ \ac{UE} can be expressed as 
\begin{equation}
\mathcal{R}_{k}=\log _{2}\left(1+\operatorname{SINR}_{k}\right),
\label{eqn_IIB_1}
\end{equation}
where,
\begin{equation}
\operatorname{SINR}_{k}=\frac{\gamma\left|\mathbf{w}_{k}\mathbf{h}_{k} \right|^{2}}{1+\gamma \sum_{j=1, j \neq k}^{K}\left| \mathbf{w}_{j}\mathbf{h}_{k} \right|^{2}}. 
\label{eqn_IIB_2}
\end{equation}
Based on \eqref{eqn_IIB_1}, the achievable sum-rate will be $R_{\mathrm{sum}}=\sum_{k=1}^{K} R_{k}$. We assume that perfect channel state information is available at the \ac{BS}. Assuming \ac{MRT}, the precoding vector $\mathbf{w}_{k}$ is given by $\mathbf{w}_{k}=\mathbf{h}_{k}^{\mathrm{H}} / \sqrt{\NBS}$. Accordingly, the achievable rate in \eqref{eqn_IIB_1} can then be written as 
\begin{equation}
\mathcal{R}_{k}=\log _{2}\left(1+\frac{\gamma \NBS}{1+\gamma \NBS \sum_{j=1, j \neq k}^{K} \mathcal{G}_{a, k j}^{2}}\right), 
\label{eqn_IIB_3}
\end{equation}
where $\mathcal{G}_{a, k j}$ denotes the value of the inner product of $\mathbf{h}_{k}$ and $\mathbf{h}_{j}$, which is defined as
\begin{equation}
\mathcal{G}_{a, k j}=\frac{\left|\mathbf{h}_{k}^{\mathrm{H}}(r_{k},\theta_{k}, \varphi_{k}) \mathbf{h}_{j}(r_{j},\theta_{j}, \varphi_{j})\right|}{\NBS}. 
\label{eqn_B4}
\end{equation}
The symbol $a \in [\mathrm{ula},\mathrm{uca}]$ denotes the array configuration, such as \ac{ULA} or \ac{UCA}. As evident from \eqref{eqn_IIB_3}, the term $\mathcal{G}_{a, kj}$ plays a critical role in determining the achievable sum-rate. With \ac{MRT} beamforming, when \acp{UE} are closely spaced, the inter-user interference is primarily governed by the overlap of the main lobes, which is determined by the beamwidth and beamdepth. In contrast, when \acp{UE} are sufficiently separated, the interference is mainly dictated by the sidelobe-induced channel correlations. In the subsequent sections, we analyze the behavior of $\mathcal{G}_{a, kj}$ in the distance domain for both \ac{ULA} and \ac{UCA} configurations.
\begin{figure}[t]
\centering
\includegraphics[width=1\columnwidth]{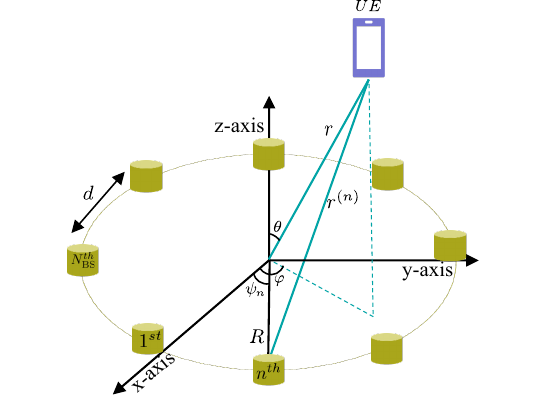}
\setlength{\belowcaptionskip}{-20pt}
\setlength{\abovecaptionskip}{0pt}
\caption{\ac{UCA} for near-field communication system. }
\label{fig1_system_model}
\end{figure}

We consider a \ac{UCA} as depicted in Fig. \ref{fig1_system_model}, where isotropic antenna elements are uniformly distributed along the circle of radius $R$. Unlike beamwidth, beamdepth is not an intrinsic electromagnetic property of an individual antenna element. Instead, it arises from the collective behavior of the antenna array. Therefore, beamdepth depends only on the array geometry and is not influenced by the radiation pattern of the individual antenna elements. The \ac{UE} is located at a distance $r$ from the center of the \ac{UCA} and subtends elevation angle $\theta$ and azimuth angle $\varphi$. In polar coordinates, the geometrical position of each antenna element is represented as $(R, \psi_n)$, where $\psi_n = \frac{2\pi n}{\NBS}$ for $n \in [1,2,\dotsc,N_\mathrm{BS}]$. The inter-element spacing in terms of arc length is $d = R \tfrac{2\pi}{\NBS}$. In this paper, the minimum distance from the \ac{BS} in the radiative near-field is set as $1.2D$ such that amplitude variations across the array are negligible, and accordingly, a \ac{USW} model is employed. Following the approach in~\cite{1137900}, it can be shown that this distance criterion applies to both \ac{ULA} and \ac{UCA} geometries. Correspondingly, the near-field array response vector is formulated as 
\begin{equation}\small
\mathbf{b}(r, \theta,\varphi) = \frac{1}{\sqrt{\NBS}} 
\begin{bmatrix}
e^{-j \frac{2\pi}{\lambda} (r^{(1)} - r)}, \dots, e^{-j \frac{2\pi}{\lambda} (r^{(\NBS)} - r)}
\end{bmatrix}^\mathrm{T},
\label{eqn_III_1}
\end{equation}
\normalsize
where $r^{(n)}$ denotes the propagation distance between the \ac{UE} and the $\nth{n}$ antenna of the \ac{UCA}. From Fig. \ref{fig1_system_model}, based on the law of cosines, $r^{(n)} = \sqrt{r^2 + R^2 - 2rR \sin \theta \cos(\varphi - \psi_n)}$, which can be further approximated using the second-order Taylor series expansion as $r^{(n)} \approx r - R \sin \theta \cos(\varphi - \psi_n) + \frac{R^2}{2r} \left( 1 - \sin^2\theta\cos^2(\varphi - \psi_n) \right)$.

\section{Analysis of UCA Near-field Beamforming} \label{Sec_lll}
In this section, we formulate array gain in the distance dimension, based on which we derive closed-form expressions for the beamdepth and \ac{EBRD} for \ac{UCA}. Finite-depth beamforming is achieved in the radiative near-field through the conventional matched filtering approach, where the phase of each radiating source is set to compensate for the path difference between the focal point $\rf$ and the source to achieve constructive interference at the focal point. Suppose we focus the beam at $\rf$; the array gain at varying distances $r \in [1.2D \ \infty]$, but same azimuth $\varphi$ and elevation angle $\theta$ is given as
{\begin{equation}\small
\GR = \left| \mathbf{w}\left(r,\theta,\varphi \right)^{\mathsf{H}} \mathbf{b}\left(\rf,\theta,\varphi \right) \right|, 
\label{eqn_beamforming}
\end{equation}
\normalsize}
where $\mathbf{w} (r, \theta,\varphi)$ is the beamforming vector. 
\begin{theorem} 
The normalized array gain obtained by near-field beamforming for a \ac{UCA} can be approximated as follows
{\begin{equation}\small
\begin{aligned}
\GR = \frac{1}{\NBS} \left| \sum_{n=1}^{\NBS} e^{j \frac{2\pi}{\lambda} \left\{\frac{R^2}{2} r_\mathrm{eff} 
\left( 1 - \sin^2 \theta \cos^2(\varphi - \psi_n) \right) \right\}} \right|
\approx \left| J_0(\zeta) \right|
\end{aligned}
\label{eqn_IIIA_1}
\end{equation}
\normalsize}
where $\zeta = \frac{\pi \RDUCA}{16} r_{\mathrm{eff}} \sin^2\theta$, $\RDUCA$ is the Rayleigh distance of the \ac{UCA}, and $r_{\mathrm{eff}} = \left| \frac{r - r_f}{r r_f} \right|$.
\label{theorem1}
\end{theorem}
\begin{proof}
The proof is provided in \textbf{Appendix A}. 
\end{proof}
\subsection{Beamdepth}
We define the beamdepth $\BD$ as the distance interval $r \in [\rf^\mathrm{min}, \rf^\mathrm{max}]$ where normalized array gain $\GR$ is at most $\unit[3]{dB}$ lower than its maximum value. 
\newtheorem{corollary}{Corollary}
\begin{corollary}
For a \ac{UCA}, the beamdepth $\BDUCA$ obtained by focusing a beam at a distance $\rf$ from the \ac{BS} is given by
\begin{equation}
\BDUCA = 
\begin{cases}
\frac{32 \pi\alpha_{\mathrm{\scalebox{0.5}{3dB}}} \RDUCA \rf^2 \sin^2{\theta}} {(\pi \RDUCA \sin^2{\theta)^2} - ({16 \rf \alpha_{\mathrm{\scalebox{0.5}{3dB}}} })^2}, 
& \rf < \frac{\pi \RDUCA}{16 \alpha_{\mathrm{\scalebox{0.5}{3dB}}}}\sin^2{(\theta)} \\[10pt]
\infty, 
& \rf \geq \frac{\pi \RDUCA}{16 \alpha_{\mathrm{\scalebox{0.5}{3dB}}}}\sin^2{(\theta)}
\end{cases}
\label{eqn_IIIA_2}
\end{equation}
\label{theorem2}
\end{corollary}
\begin{proof}
The proof is provided in \textbf{Appendix B}. 
\end{proof}
Note that although $\BDUCA$ in \eqref{eqn_IIIA_2} is independent of azimuth angle $\varphi$, it depends on the elevation angle $\theta$. Furthermore, $\BDUCA$ increases as the focus distance $\rf$ increases or the elevation angle moves towards the boresight direction ($\theta = 0^{\circ}$).

\subsection{Effective Beamfocusing Rayleigh Distance (EBRD)} 
Finite depth beams are achieved in the near-field, only when the focus distance $\rf$ lies within a certain distance limit for a given elevation angle $\theta$ in \eqref{eqn_IIIA_2}. We derive this limit for a \ac{UCA} and refer to it as \ac{EBRD}. 
\begin{corollary}
The farthest distance at which beamfocusing for a \ac{UCA} can be achieved is $\rf < \frac{\pi \RDUCA}{16 \alpha_{\mathrm{\scalebox{0.5}{3dB}}}}\sin^2{(\theta)}$.
\label{theorem3}\end{corollary}
\begin{proof}
The proof is provided in \textbf{Appendix C}. 
\end{proof}
\ac{EBRD} is minimum in the boresight direction ($\theta=0^\circ$) and increases at end-fire directions ($\theta=\pm 90^\circ$). The near-field codebook proposed in~\cite{10243590} can be extended to three-dimensional scenarios by utilizing beamdepth and the \ac{EBRD}. Similar to~\cite{10934779}, for a given pair of azimuth and elevation angles, multiple range samples are generated. Starting from a minimum distance, these range samples are generated with beamdepth-based spacing until the \ac{EBRD} limit is reached. This sampling procedure is then repeated for all angular directions to construct the complete 3D codebook.


\section{Comparative analysis with ULA} \label{Sec_IV}
In this section, we investigate the beamdepth and the \ac{EBRD} to characterize spatial correlation in the distance domain for both \ac{ULA} and \ac{UCA}.
\subsection{Beamdepth and EBRD}
A narrow beamdepth and extended \ac{EBRD} facilitate reduced inter-user interference, thereby enhancing spatial multiplexing. In this subsection, we compare the spatial correlation of the \ac{ULA} and \ac{UCA} in terms of beamdepth and forelobes under the constraint of fixed antenna element count and fixed aperture length. First, we compare the aperture lengths of the arrays under the constraint of equal antenna element count.
\begin{theorem} 
For equal antenna count, the $ \ac{UCA} $ has an aperture length reduced by a factor of $ \pi $ relative to the $ \ac{ULA} $.
\label{theorem4}
\end{theorem}
\begin{proof}
For $\NBS$ antenna elements, the aperture length of a \ac{ULA} is $\DULA \approx \NBS d$. In contrast, the antenna elements in a \ac{UCA} are placed on a circle of circumference $\pi \DUCA =\NBS d$, resulting in an effective aperture length of $\DUCA = \frac{\NBS d}{\pi}$. Therefore, $\DUCA = \frac{\DULA}{\pi}$, which completes the proof. Also $\RDULA=\frac{2\DULA^2}{\lambda}$, exceeds $\RDUCA=\frac{2\DUCA^2}{\lambda}$ by a factor of $\pi^2$.
\end{proof}

The \ac{ULA} gain function is given by $\mathcal{G}_{\mathrm{ula}} \approx \left| \frac{C^2(\gamma) + S^2(\gamma)}{\gamma^2} \right|$, where $C(\gamma)$ and $S(\gamma)$ denote the Fresnel integrals. The parameter $\gamma$ is defined as $\gamma = \sqrt{\frac{N^2 d^2 \cos^2(\varphi)}{2\lambda}r_{\mathrm{eff}}}$. 
Furthermore, $\alpha_{\mathrm{\scalebox{0.5}{3dB}}} \stackrel{\Delta}{=} \left\{ \gamma \mid \left|\mathcal{G}_{\mathrm{ula}}(\gamma)\right| = 0.5 \right\}$. The beamdepth $\BDULA$ and the corresponding \ac{EBRD} for a \ac{ULA} are given by \cite{10934779}
\begin{equation}\small
\BDULA = 
\begin{cases}
\displaystyle \frac{8 \alpha_\mathrm{\scalebox{.5}{3dB}}  \RDULA \rf^2 \cos^2{\varphi}}{ \left( \RDULA \cos^2{\varphi} \right)^2 - \left(4 \alpha_\mathrm{\scalebox{.5}{3dB}} \rf \right)^2}, 
& \rf < \displaystyle\frac{\RDULA}{4 \alpha_\mathrm{\scalebox{.5}{3dB}}} \cos^2{\varphi} \\[10pt]
\infty, 
& \rf \geq \displaystyle\frac{\RDULA}{4 \alpha_\mathrm{\scalebox{.5}{3dB}}} \cos^2{(\varphi)}
\end{cases}
\label{eqn11}
\end{equation}
\normalsize
where $\frac{\RDULA}{4 \, \alpha_\mathrm{\scalebox{.5}{3dB}}} \cos^2{\varphi}$ represents the \ac{EBRD}, with $\RDULA$ denoting the Rayleigh distance of the \ac{ULA}. Furthermore, $\varphi$ is the angle between a line perpendicular to the array face and the array axis.
 Notably, the \ac{EBRD} for a \ac{ULA} is maximized in the boresight direction, i.e., at $\varphi = 0^{\circ}$ in \eqref{eqn11}, whereas the \ac{UCA} exhibits its minimum \ac{EBRD} at $\theta=0^{\circ}$ (boresight direction), as given in Corollary~\ref{theorem3}. Consequently, \acp{UE} located at the boresight of the \ac{UCA} or in the endfire direction of the \ac{ULA} may not fully benefit from near-field beamfocusing. Note that the rotational symmetry of a \ac{UCA} is confined to the azimuth plane and does not extend to the elevation dimension. Consequently, the beamfocusing characteristics of a \ac{UCA} remain invariant with respect to the azimuth but depend on the elevation angle.


\begin{figure*}[!t]
  \centering
  \begin{minipage}[b]{0.32\textwidth}
    \centering
    \includegraphics[width=\textwidth]{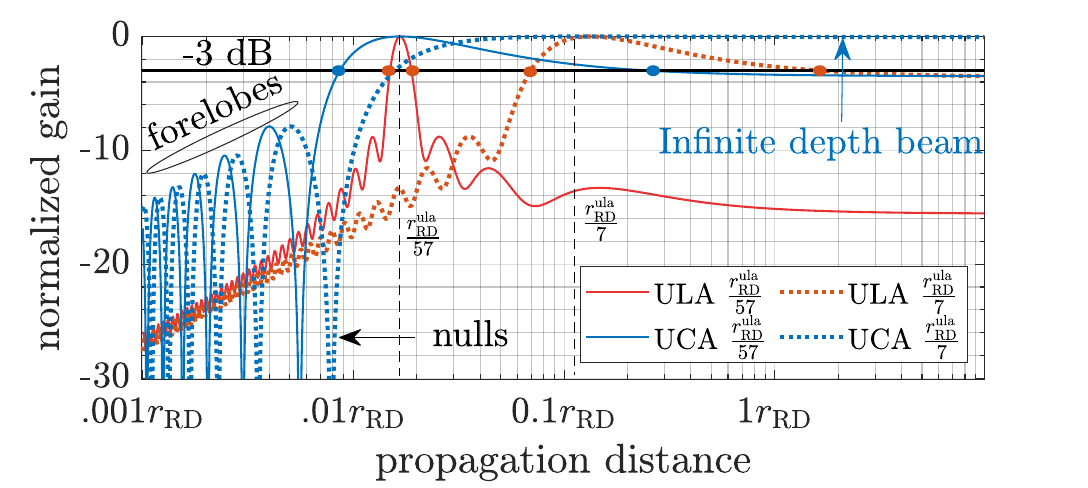}
    \caption{Beamdepth comparison for fixed $\NBS=256$.}
    \label{fig4_BD_same_N.pdf}
  \end{minipage}%
  \hspace{0.1mm}
  \begin{minipage}[b]{0.32\textwidth}
    \centering
    \includegraphics[width=\textwidth]{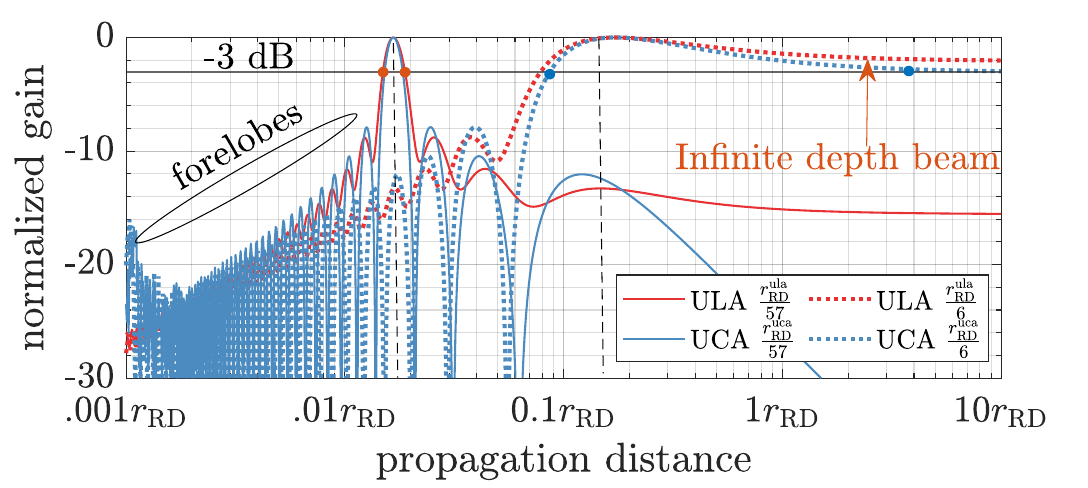}
    \caption{ Beamdepth comparison for fixed aperture length of $\unit[1.36]{m}$.}
    \label{fig3_BD_same_aperture}
  \end{minipage}
  \hspace{.1mm}
  \begin{minipage}[b]{0.32\textwidth}
    \centering
    \includegraphics[width=\textwidth]{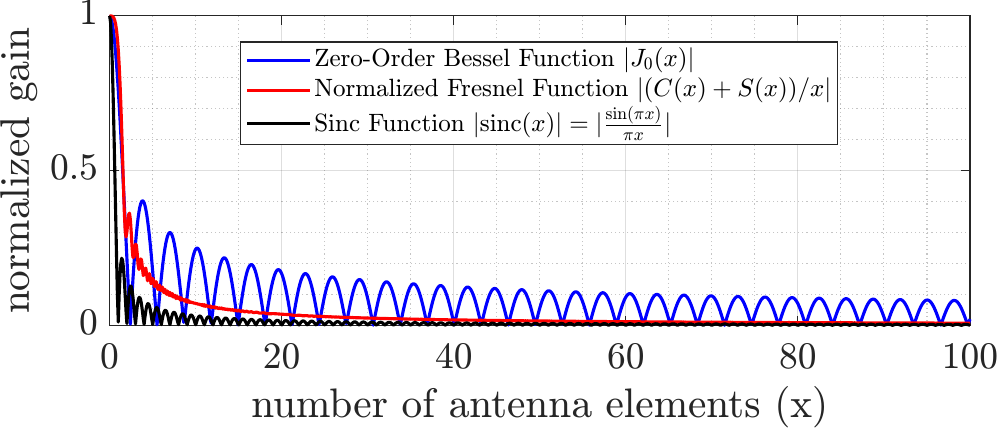}
    \caption{Comparison of array gain functions for \ac{ULA} vs \ac{UCA}.}
    \label{fig2_Bessel_vs_Fresnel}
  \end{minipage}
\end{figure*}


\subsubsection{Fixed Number of Antenna Elements}\label{IV-A_1}
\textbf{Under a fixed antenna count $\NBS$, the \ac{ULA} achieves a narrower beamdepth and a more extended \ac{EBRD} compared to the \ac{UCA}}. When comparing the beamdepth at the same focus distance, $\RD$ is the only significant geometry-dependent parameter in \eqref{eqn_IIIA_2} and \eqref{eqn11}. Both expressions indicate that beamdepth decreases with increasing $\RD$. As stated in Theorem \ref{theorem4}, since $\DULA > \DUCA$, it follows that $\RDULA > \RDUCA$, leading to a narrower beamdepth for the ULA compared to the UCA. To further illustrate this, Fig.~\ref{fig4_BD_same_N.pdf} shows the normalized array gain as a function of propagation distance. The distance interval between the 3\ dB gain points defines the beamdepth. For a fair comparison, we choose $\theta = 90^\circ$ for the \ac{UCA} and $\varphi = 0^\circ$ for the \ac{ULA}, as these angles minimize the beamdepth according to \eqref{eqn_IIIA_2} and \eqref{eqn11}, respectively. When focusing the beam at a near-field distance of $\rf = \RDULA/57$, the ULA achieves a significantly narrow beamdepth of $\unit[1.4]{m}$, compared to $\unit[84]{m}$ for the UCA. Here $\RDULA = \unit[348]{m}$ and $\RDUCA = \unit[35]{m}$. Furthermore, \ac{ULA} achieves finite depth beams up till (\ac{EBRD}) $\RDULA/7 =\unit[49.7]{m}$; the same limit is also obtained from \eqref{eqn11}, where $\alpha_\mathrm{\scalebox{.5}{3dB}} =1.75$ for $\varphi=0^\circ$. On the other hand \ac{EBRD} for \ac{UCA} is only $\RDULA/57=\RDUCA/6 \approx \unit[6]{m}$. As shown in Fig.~\ref{fig4_BD_same_N.pdf}, ULA maintains a finite beamdepth up to $\RDULA/7$, while it becomes infinite for \ac{UCA} highlighted by the dotted blue line for \ac{UCA} at $\RDULA/7$, which stays above the $\unit[3]{dB}$ threshold beyond the focusing point. Fig.~\ref{fig4_BD_same_N.pdf} also illustrates the spatial correlation between a \ac{UE} located at the focal point $\rf$ and \acp{UE} positioned outside the interval $r \in [\rf^\mathrm{min}, \rf^\mathrm{max}]$. \textbf{Notably, the forelobes of the \ac{UCA} exceed those of the \ac{ULA} by approximately $\unit[8$–$10]{dB}$, indicating higher spatial correlation for \ac{UCA} between the \ac{UE} at $\rf$ and \acp{UE} outside the beamdepth region}. This elevated correlation may adversely impact the capacity performance of the \ac{UCA}. Conversely, the \ac{UCA} can form nulls in the spatial correlation pattern, which are not achievable with the \ac{ULA}, as also shown in Fig.~\ref{fig4_BD_same_N.pdf}. These nulls may be desirable in adaptive beamforming and secure communication.

\subsubsection{Fixed Aperture Length}
\textbf{For a fixed aperture length, the \ac{UCA} demonstrates a slightly narrower beamdepth and a marginally larger \ac{EBRD} compared to the \ac{ULA}}. Achieving equal aperture lengths for both array configurations in Fig.~\ref{fig3_BD_same_aperture} requires a significantly higher number of antenna elements for the \ac{UCA}. For instance, at a carrier frequency of \unit[28]{GHz}, matching aperture lengths of $\unit[1.36]{m}$ necessitate 256 elements for the \ac{ULA}, whereas the \ac{UCA} requires 801 elements. In this case, $\DULA = \DUCA$ implies $\RDULA = \RDUCA$, and the slightly narrower beamdepth for the \ac{UCA} in \eqref{eqn_IIIA_2} compared to \ac{ULA} in \eqref{eqn11} arises from differences in $\alpha_{\mathrm{\scalebox{0.5}{3dB}}}$ and constant factors. As shown in Fig.~\ref{fig3_BD_same_aperture}, when the beam is focused at $\RD^\mathrm{ula}/57$, the \ac{UCA} achieves a slightly narrower beamdepth compared to the \ac{ULA}. In the boresight scenario, the \ac{EBRD} of the \ac{UCA} is approximately $\RDUCA/6$, which slightly exceeds that of the \ac{ULA}, $\RDULA/7$. The \ac{EBRD} for \ac{UCA} can also be derived directly from Corollary~\ref{theorem3}, using $\alpha_\mathrm{\scalebox{.5}{3dB}} = 1.2$. At $\RDUCA/6$, the \ac{UCA} maintains a well-defined finite-depth beam, whereas the \ac{ULA} exhibits an infinite-depth beam—evident from the trailing edge at $\RDULA/6$ (red dotted curve) remaining above the 0.5 gain threshold. Additionally, despite having the same aperture length, the forelobes of the \ac{UCA} exceed those of the \ac{ULA} by approximately $\unit[1$--$2]{dB}$ across all distances. As a result, the \ac{UCA} exhibits stronger spatial correlation between a \ac{UE} at $\rf$ and surrounding \acp{UE} outside the beamdepth region.


\subsection{Asymptotic Orthogonality}
Channels $\mathbf{h}_k$ and $\mathbf{h}_j$ are asymptotically orthogonal if $\displaystyle \lim_{N_\mathrm{BS} \to \infty} \frac{|\mathbf{h}_k^{\mathrm{H}} \mathbf{h}_j|}{N_\mathrm{BS}} \to 0$. 
For a \ac{ULA}, the distance-domain array gain is given in \eqref{eqn11}. In the limit as $N_\mathrm{BS} \to \infty$, $\gamma \to \infty$, and using the property $C(\infty) = S(\infty) = 0.5$, it follows that $\GU \to 0$. For a \ac{UCA}, the array gain in the distance domain is characterized by the zero-order Bessel function $J_0(\zeta)$ as shown in~\eqref{eqn_IIIA_1}. The asymptotic form $J_0(\zeta) \sim \sqrt{\frac{2}{\pi \zeta}} \cos\left(\zeta - \frac{\pi}{4}\right)$~\cite{olver1954asymptotic} implies $J_0(\zeta) \to 0$ as $\zeta \to \infty$. Since $\zeta$ is proportional to the array radius $R$, increasing $N_\mathrm{BS}$ increases $R$ and consequently $\zeta$, ensuring $\mathcal{G}_{\mathrm{uca},jk} \to 0$ for \ac{UCA} as well. The angular-domain array gain for the \ac{ULA} is governed by the $\mathrm{sinc}$ function, while for the \ac{UCA}, it is defined by the Bessel function. The asymptotic orthogonality of both array types in angular and distance domains is established in the literature. For further comparison, Fig.~\ref{fig2_Bessel_vs_Fresnel} illustrates the decay behavior of these array gain functions, with $x$ denoting a scaled variable proportional to $N_\mathrm{BS}$. \textbf{While both \ac{ULA} and \ac{UCA} achieve asymptotic orthogonality, their array gain functions exhibit distinct decay profiles: the $\mathrm{sinc}$ and Fresnel integrals decay monotonically as $1/x$, whereas the Bessel function decays as $1/\sqrt{x}$ with oscillatory behavior}. Consequently, for a given $N_\mathrm{BS}$, the $\mathrm{sinc}$ function yields the lowest correlation, followed by the Fresnel and Bessel functions. This indicates that \ac{ULA} multiplexes a higher number of \acp{UE} in the angular domain, followed by its distance domain, and then by the distance/angle domain of the \ac{UCA}.

\section{Simulation Results} \label{Sec_V}
In this section, we conduct Monte Carlo simulations to compare the achievable sum-rate of the \ac{UCA} with other array configurations. The primary benchmark is the \ac{ULA}, which employs a one-dimensional array geometry. In addition, we include comparisons with the \ac{URA}, since the \ac{UCA} provides spatial resolution in both azimuth and elevation dimensions. The simulations are performed at a carrier frequency of $\unit[30]{GHz}$ and an \ac{SNR} of $\unit[15]{dB}$. The \acp{UE} are uniformly distributed in distance, azimuth, and elevation, i.e., $r \sim \mathcal{U}[1.2D, \EBRD]$, $\varphi \sim \mathcal{U}[-\pi/2, \pi/2]$, and $\theta \sim \mathcal{U}[-\pi/2, \pi/2]$. We consider a multipath channel model consisting of $L=5$ N\ac{LoS} components. The Rician factor is set to $\kappa = 10$.

Fig.~\ref{fig5_sumrate_same_N.pdf} illustrates the achievable sum-rate as a function of the number of users. We compare \ac{UCA} and \ac{URA} configurations under a fixed number of antenna elements, $\NBS = 256$. Specifically, the considered \ac{URA} geometries are $\{1\times256,\;2\times128,\;4\times64,\;16\times16\}$. The corresponding aperture lengths, determined by the array diagonal for \acp{URA}, are $\unit[1.27]{m}$, $\unit[0.63]{m}$, $\unit[0.32]{m}$, and $\unit[0.11]{m}$, respectively, while the \ac{UCA} has an aperture length of $\unit[0.40]{m}$. Among the \ac{URA} configurations, elongated geometries such as the \ac{ULA} exhibit the largest aperture, whereas the square \ac{URA} yields the smallest. As shown in Fig.~\ref{fig5_sumrate_same_N.pdf}, the \ac{ULA} achieves the highest sum-rate due to its larger aperture length, which results in narrower beamdepth and an extended \ac{EBRD}. In contrast, the \ac{UCA} attains the lowest sum-rate due to wider beamdepth and higher forelobe levels. Among the \ac{URA} configurations, the square \ac{URA} yields the lowest sum-rate, consistent with its wider beamdepth and reduced \ac{EBRD} caused by the smaller aperture. The array gain of a \ac{URA} in the distance domain is governed by a two-dimensional Fresnel function \cite{10934779}, while its angular response in azimuth and elevation follows sinc functions. Consequently, the beam pattern analysis developed for the \ac{ULA} extends to \ac{URA} configurations. As a result, the \ac{UCA} exhibits higher forelobe levels compared to \acp{URA}, leading to a degraded sum-rate. Under certain \ac{UE} distributions, such as when \acp{UE} are located near the endfire directions, a \ac{UCA} may achieve a higher sum-rate since its beamdepth remains invariant with respect to the azimuth angle.

\begin{figure}[t]
\centering
\includegraphics[width=1\linewidth]{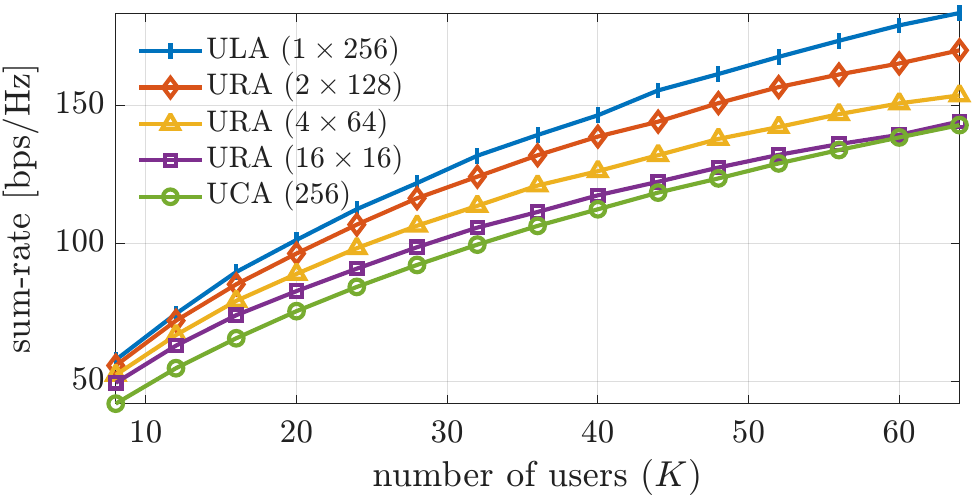}
\caption{Sum-rate comparison for fixed antenna, $\NBS=256$.}
\label{fig5_sumrate_same_N.pdf}
\end{figure}

\begin{figure}[t]
\centering
\includegraphics[width=1\columnwidth]{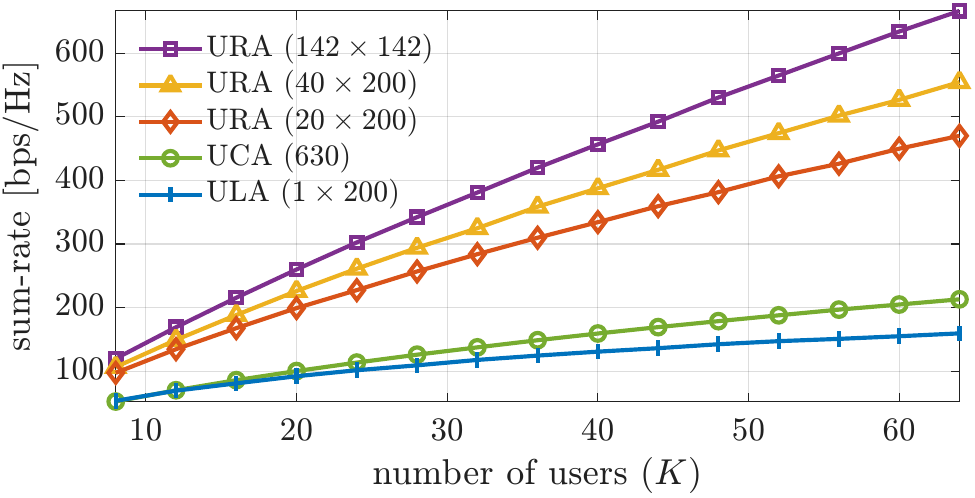}
\caption{Sum-rate comparison for fixed aperture length, $D = \unit[1]{m}$.}
\setlength{\belowcaptionskip}{-20pt}
\setlength{\abovecaptionskip}{0pt}
\label{fig6_sumrate_same_aperture}
\end{figure}

Fig.~\ref{fig6_sumrate_same_aperture} compares the achievable sum-rate under a fixed aperture length of $\unit[1]{m}$. Under this constraint, the square \ac{URA} achieves the highest sum-rate performance. When the aperture length is fixed, the influence of array geometry on beamdepth and \ac{EBRD} is significantly reduced, and the beamforming gain associated with the number of antenna elements becomes the dominant factor. Among the considered geometries, the square \ac{URA} accommodates the largest number of antenna elements for a given aperture length, thereby providing the highest beamforming gain and the lowest forelobe levels among the \ac{URA} configurations. 
In contrast, the \ac{ULA} yields the lowest sum-rate since it accommodates the fewest antenna elements under the same aperture constraint. Notably, the sum-rate difference between the \ac{ULA} and the \ac{UCA} remains marginal. For high user loading, the \ac{UCA} achieves a slightly higher sum-rate due to its marginally narrower beamdepth compared to the \ac{ULA} under a fixed aperture length.

\section{Conclusion}
In this paper, we compared the spatial correlation between \acp{UE} for \ac{ULA} and \ac{UCA} configurations. Our results indicate that the capacity performance of the \ac{UCA} is limited due to high spatial correlation, despite the 2D array geometry. As a direction for future work, it would be interesting to analyze and compare the spatial correlation and corresponding capacity of uniform rectangular arrays and uniform concentric circular arrays in the near-field.

\bibliographystyle{IEEEtran}
\bibliography{IEEEabrv,my2bib}

\begin{thebibliography}{1}
\providecommand{\url}[1]{#1}
\csname url@samestyle\endcsname
\providecommand{\newblock}{\relax}
\providecommand{\bibinfo}[2]{#2}
\providecommand{\BIBentrySTDinterwordspacing}{\spaceskip=0pt\relax}
\providecommand{\BIBentryALTinterwordstretchfactor}{4}
\providecommand{\BIBentryALTinterwordspacing}{\spaceskip=\fontdimen2\font plus
\BIBentryALTinterwordstretchfactor\fontdimen3\font minus \fontdimen4\font\relax}
\providecommand{\BIBforeignlanguage}[2]{{%
\expandafter\ifx\csname l@#1\endcsname\relax
\typeout{** WARNING: IEEEtran.bst: No hyphenation pattern has been}%
\typeout{** loaded for the language `#1'. Using the pattern for}%
\typeout{** the default language instead.}%
\else
\language=\csname l@#1\endcsname
\fi
#2}}
\providecommand{\BIBdecl}{\relax}
\BIBdecl

\bibitem{10934779}
A.~Abdallah, A.~Hussain, A.~Celik, and A.~M. Eltawil, ``Exploring frontiers of polar-domain codebooks for near-field channel estimation and beam training: A comprehensive analysis, case studies, and implications for 6{G},'' \emph{IEEE Signal Process. Mag.}, vol.~42, no.~1, pp. 45--59, 2025.

\bibitem{10988573}
A.~Hussain, A.~Abdallah, and A.~M. Eltawil, ``Redefining polar boundaries for near-field channel estimation for ultra-massive {MIMO} antenna array,'' \emph{IEEE Trans. Wireless Commun.}, vol.~24, no.~10, pp. 8193--8207, 2025.

\bibitem{10243590}
Z.~Wu, M.~Cui, and L.~Dai, ``Enabling more users to benefit from near-field communications: From linear to circular array,'' \emph{IEEE Trans. Wireless Commun.}, vol.~23, no.~4, pp. 3735--3748, 2024.

\bibitem{10541333}
M.~Cui and L.~Dai, ``Near-field wideband beamforming for extremely large antenna arrays,'' \emph{IEEE Trans. Wireless Commun.}, vol.~23, no.~10, pp. 13\,110--13\,124, 2024.

\bibitem{11428208}
A.~Hussain, A.~Abdallah, A.~Celik, E.~Björnson, and A.~M. Eltawil, ``Analyzing {URA} geometry for enhanced near-field beamfocusing and spatial degrees of freedom,'' \emph{IEEE Trans. Commun.}, pp. 1--1, 2026.

\bibitem{1137900}
J.~Sherman, ``Properties of focused apertures in the fresnel region,'' \emph{IRE Transactions on Antennas and Propagation}, vol.~10, no.~4, pp. 399--408, 1962.

\bibitem{olver1954asymptotic}
F.~W. Olver, ``The asymptotic expansion of bessel functions of large order,'' \emph{Philos. Trans. R. Soc. Lond. A}, vol. 247, no. 930, pp. 328--368, 1954.

\end{thebibliography}
\appendices
\section*{Appendix A: Proof of Theorem \ref{theorem1}} 
\label{Appendix_A}
Utilizing \eqref{eqn_III_1} and \eqref{eqn_beamforming}, the array gain in the distance domain for \ac{UCA} is given by
$\mathcal{\GR} = \frac{1}{\NBS} \left| \sum_{n=1}^{\NBS} e^{j \frac{2\pi}{\lambda} \left( \frac{R^2}{2} r_\mathrm{eff} \left( 1 - \sin^2\theta \cos^2(\varphi - \psi_n) \right) \right)} \right|$,
where $r_\mathrm{eff} = \left| \frac{r - \rf}{r \rf} \right|$.
Defining $\zeta = \frac{\pi}{\lambda} \frac{R^2}{2} r_\mathrm{eff} \sin^2\theta$, substituting $R = \DUCA/2$ and $\RDUCA = \frac{2\DUCA^2}{\lambda}$ yields $\zeta = \frac{\pi \RDUCA}{16} r_\mathrm{eff} \sin^2\theta$.
Without loss of generality, let $\varphi = 0$ due to rotational symmetry of \ac{UCA}, so
$\mathcal{\GR} \approx \frac{1}{\NBS} \left| \sum_{n=1}^{\NBS} e^{-2j \zeta \cos^2(\psi_n)} \right|$.
Using trigonometric identity $\cos(2x) = 2\cos^2(x) - 1$, the gain simplifies to
$\mathcal{\GR} \approx \frac{1}{\NBS} \left| \sum_{n=1}^{\NBS} e^{-j \zeta \cos(2\psi_n)} \right|$.
For large $\NBS$, this sum approximates the integral
$\frac{1}{2\pi} \int_0^{2\pi} e^{-j \zeta \cos(\psi_n)} d\psi_n = J_0(\zeta)$,
where $J_0(\zeta)$ is the Bessel function of the first kind and order zero. The exact array gain for a \ac{UCA} can be expressed as a sum of Bessel functions of all integer orders. Using the Jacobi--Anger expansion \cite{10243590}, it can be shown that for $\zeta \ll \frac{\NBS}{2}$, the zeroth-order Bessel function dominates, while the contributions of the higher-order terms are negligible. Hence, $\mathcal{\GR} \approx |J_0(\zeta)|$, which completes the proof.
\section*{Appendix B: Proof of Corollary \ref{theorem2}} 
\label{Appendix_B}
To obtain the distance points where the gain function in \eqref{eqn_IIIA_1} equals $\unit[3]{dB}$ of its maximum value, we define $\alpha_{\mathrm{\scalebox{0.5}{3dB}}} \stackrel{\Delta}{=}\left\{ \zeta \mid \abs{\GR\left(\zeta \right)}^2 = 0.5 \right\}$. Thus, 
$\alpha_{\mathrm{\scalebox{0.5}{3dB}}} = \frac{\pi \RDUCA}{16} r_{\mathrm{eff}} \sin^2 \theta$, where $r_{\mathrm{eff}} = \left| \frac{r - \rf}{r \rf} \right|.$ Solving for $r$ yields two solutions: 
$r = \frac{ \pi \RDUCA \sin^2 \theta \rf }{ \pi \RDUCA \sin^2 \theta \pm 16 \alpha_{\mathrm{\scalebox{0.5}{3dB}}} \rf }
$. Therefore, $\rf^\mathrm{max} = \frac{ \pi \RDUCA \sin^2 \theta \rf }{ \pi \RDUCA \sin^2 \theta - 16 \alpha_{\mathrm{\scalebox{0.5}{3dB}}} \rf }$, and $\rf^\mathrm{min} = \frac{ \pi \RDUCA \sin^2 \theta \rf }{ \pi \RDUCA \sin^2 \theta + 16 \alpha_{\mathrm{\scalebox{0.5}{3dB}}} \rf }$. The distance window between $\rf^\mathrm{max}$ and $\rf^\mathrm{min}$ is the interval where $\GR$ is less than or equal to $\unit[3]{dB}$. Therefore, $\BD=\rf^\mathrm{max}-\rf^\mathrm{min}$ is given by (\ref{eqn_IIIA_2}), which completes the proof. 
\section*{Appendix C: Proof of Corollary \ref{theorem3}}\label{Appendix_C}
In \textit{Theorem} (\ref{theorem1}), the maximum value of $\BDUCA$ is obtained when the factor in the denominator ${(\pi \RDUCA \sin^2{\theta)^2} - ({16 \rf \alpha_{\mathrm{\scalebox{0.5}{3dB}}} })^2} = 0$. Thus, the farthest angle-dependent axial distance $\rf$, where we can have finite depth beamforming is less than $\frac{\pi\RDUCA}{16 \alpha_{\mathrm{\scalebox{0.5}{3dB}}}}\sin^2{(\theta)}$.
\normalsize
\end{document}